\documentclass[%
 reprint,
 amsmath,amssymb,
 aps,
]{revtex4-2}

\usepackage{graphicx}
\usepackage{dcolumn}
\usepackage{bm}
\usepackage{amsmath}
\usepackage{amsfonts}
\usepackage{amsthm}
\newtheorem{theorem}{Theorem}
\newtheorem{corollary}[theorem]{Corollary}
\newtheorem{lemma}[theorem]{Lemma}

\newtheorem{definition}[theorem]{Definition}
\newtheorem{prop}[theorem]{Proposition}

\usepackage[hidelinks]{hyperref}
\usepackage{physics}
\usepackage{color}
\usepackage{float}

\begin{document}

\preprint{APS/123-QED}

\title{Distributed Monogamy of Entanglement limits Quantum Channel Simulation}

\author{Rabsan Galib Ahmed$^{1,2}$}
 \email{rgahmed@uwaterloo.ca}
\author{Graeme Smith$^{1,2}$}%
 \email{graeme.smith@uwaterloo.ca}
\affiliation{$^1$Department of Applied Mathematics, University of Waterloo, Ontario N2L 3G1, Canada.\\
$^2$Institute for Quantum Computing, University of Waterloo, Ontario N2L 3G1, Canada.
}%




\date{\today}

\begin{abstract}
      Entanglement is monogamous: if it is shared among more than two parties, the entanglement between any pair cannot be very strong. For an integer $k\geq 2$, $k$-extendibility of a state $\rho_{AB}$ quantifies this as the number of copies of $B$ that can be simulated by the state's environment. We introduce \textit{fractional extendibility}, which gives a finer characterization of the quantum correlation that is leaked to the environment, and prove that it is invariant under tensor products and monotonic under local processing.  We also establish the \textit{distributed monogamy of entanglement}: for any state on $AB_1B_2\dots B_n$, the maximum average probability of extracting an EPR pair from a random subset of $k \leq n/2$ systems among the $B_i$'s is the fraction $k/n$. With these tools we resolve a conjecture posed by Matthew Hastings: any quantum erasure channel with erasure probability at least $\frac{1}{2}$ cannot simulate a less noisy erasure channel, even with asymptotically many uses of the noisier channel. 
\end{abstract}

\maketitle
Information theory can be viewed as the study of resource interconversion. In quantum Shannon theory, quantum channel simulation is a central instance of this viewpoint. Given many independent uses of a channel, one asks whether it is possible to approximate a target channel with suitable encoding and decoding operations, and what resource cost is required to do so. 
When the target channel is a noiseless quantum channel, the problem becomes determining the optimal rate at which reliable quantum communication can be extracted from a noisy channel, which is called the quantum capacity.
More generally, the problem of simulating one noisy channel with another is substantially more subtle~\cite{Bennett_2014,Berta_2011,Fang_2020}. This difficulty makes channel simulation a natural framework for probing the structure and limitations of quantum information-processing.

The quantum erasure channel, $\mathcal{N}_{\lambda}$, is a particularly important target for channel simulation.  This channel transmits its input perfectly with \textit{transmission probability} $\lambda$, while with probability $1-\lambda$ it replaces the input by an orthogonal erasure flag~\cite{Bennett_1997}. The quantum capacity of $\mathcal{N}_{\lambda}$ is given by $\mathcal{Q}(\mathcal{N}_{\lambda}) = \max\{(2\lambda-1)\log_2d_{input},0\}$~\cite{Bennett_1997}. The vanishing of the quantum capacity for $\lambda \leq 1/2$ can be understood from no-cloning: in this regime, the environment receives at least as much information about the input as the receiver. If the receiver could decode an arbitrary encoded quantum state with asymptotically unit fidelity, then the environment could do so as well, which would violate the no-cloning theorem~\cite{Wootters1982-cz}.


Quantum channels with vanishing quantum capacity exhibit richer structure than their classical counterparts, making their communication value highly non-trivial to characterize. For example, with \textit{superactivation of quantum capacity} two channels, each with zero quantum capacity, can be combined to get positive capacity. In particular, a private Horodecki channel~\cite{Horodecki_2008} together with a symmetric channel~\cite{Smith_2008}, such as the $50\%$ erasure channel, gives a joint channel with positive quantum capacity~\cite{Simth-Yard_2008}. Such structures impose strong constraints on channel simulation even within zero capacity channels. For instance, the private channel can transmit secret key at a positive rate, whereas the symmetric channel has zero private capacity. This rules out simulating the former using the latter. This is in stark contrast with classical Shannon theory, where channels can be divided into two equivalence classes under channel simulation: channels with positive capacity and channels with zero capacity. Channels within each class can all simulate each other.  

Even among quantum erasure channels $\mathcal{N}_\lambda $ with $\lambda\leq 1/2$, all of which have zero quantum and private capacities, channel simulation can still be impossible. In~\cite{Hastings2008}, Hastings raised the question: \textit{Can we simulate $\mathcal{N}_\lambda$ with $\mathcal{N}_{\gamma}$ for some $\gamma <\lambda \leq \frac{1}{2}$?} and conjectured that the answer is no. 

One special case of this conjecture can be proved directly. For any $(l+1)^{-1}\leq \gamma \leq l^{-1}$ and ${(k+1)}^{-1}< \lambda\leq k^{-1}$ with integers $l,k$ and $l>k\geq 2$, it is impossible to simulate $\mathcal{N}_\lambda$ with $\mathcal{N}_\gamma$. This is because such $\mathcal{N}_\gamma$ and $\mathcal{N}_\lambda$ are $l$-extendible and not $l$-extendible respectively. Monotonicity of extendibility under tensor power, pre- and post-processing and the compactness of the set of all $k$-extendible channels with fixed finite input-output dimension, then rules out the simulation. However, even though partial progress has been made in~\cite{Hastings2008,Alhejji2023}, the more general question remains unanswered. 

The central obstacle to answering this question is the lack of a clean characterization of exactly how much \textit{noisier} $\mathcal{N}_{\gamma}$ is compared to $\mathcal{N}_{\lambda}$ when $\gamma <\lambda \leq 1/2$. Extendibility is too crude to make this distinction when $(k+1)^{-1}\leq \gamma<\lambda \leq k^{-1}$ for some integer $k$. Therefore, we generalize $k$-extendibility to \textit{fractional extendibility} to give a finer quantification of the quantum correlation shared with the environment through a noisy channel. We prove that fractional extendibility is invariant under tensor power and monotonic under pre-processing and post-processing. Furthermore, we establish the \textit{distributed monogamy of entanglement}, a robust universal limit to the maximum average probability with which a random $k$-subset of an $n$-party state can obtain an EPR pair with some reference system.  This holds for $2k\leq n$. We upper bound the fractional extendibility of a quantum erasure channel using this limit. Finally, with all these tools we show that the aforementioned simulation is impossible for all $\gamma<\lambda\leq \frac{1}{2}$.

Note that fractional extendibility for Gaussian states and channels was introduced by the authors in an earlier article~\cite{ahmed2026}, and is generalized to arbitrary quantum systems in this work.
\begin{figure}
    \centering
    \includegraphics[width=0.9\linewidth]{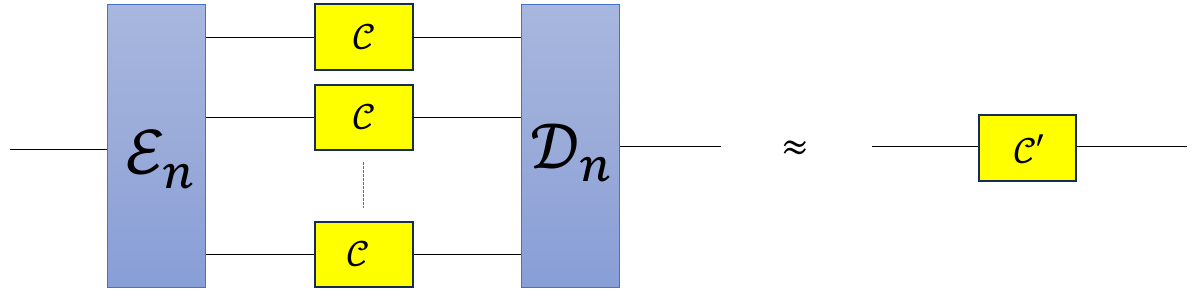}
    \caption{Simulating one quantum channel ($\mathcal{C}'$) with another ($\mathcal{C}$). We use the trace distance between the Choi states of the channels in the figure to quantify the error of simulation. In the asymptotic limit of many uses of $\mathcal{C}$ ($n\to \infty$), we require a vanishingly small error.}
    \label{fig:channel simulation}
\end{figure}

\textit{Preliminaries ---} We say that a quantum channel, $\mathcal{C}: A \to B$, simulates (Fig.~\ref{fig:channel simulation}) the quantum channel $\mathcal{C}': A'\to B'$ if there is a sequence of encoders, $\mathcal{E}_n: A'\to A^{\otimes n}$ and a corresponding sequence of decoders, $\mathcal{D}_n: B^{\otimes n} \to B'$ such that $\lim_{n\to \infty} \frac{1}{2}\norm{\mathcal{J}_{\mathcal{C}'} - \mathcal{J}_{\mathcal{D}_n\circ\mathcal{C}^{\otimes n}\circ\mathcal{E}_n}}_1 = 0$, where $\mathcal{J}_\mathcal{N}:=\mathrm{id}_R\otimes\mathcal{N}(\phi_{RA'})$ is the Choi state of the channel, $\mathcal{N}:A'\to B'$ and $R\cong A', d=\dim R$ with $\phi_{RA'}=\frac{1}{d} \sum_{i,j=1}^d \ketbra{ii}{jj}$, the maximally entangled state~\footnote{Equivalently we could use the diamond distance. As we consider channels with fixed finite dimensional input and output, vanishing of these two distances are equivalent.}. The Choi state of an erasure channel, $\mathcal{N}_{\lambda}$, is given by
\begin{align}
    \rho^{\lambda}_{RB} = \lambda \;\phi_{RC} + (1-\lambda) \frac{\mathbb{I}_R}{d}\otimes \ketbra{e}{e}_F,
\end{align}
with $B\cong C\oplus F$.

Graph theoretic methods are useful for several branches of quantum information theory~\cite{CSW2014,DSW_2013,Duan_2016,Kunjwal_2019}. A graph, $G\equiv(V,E)$, comprises of a vertex set, $V$ and an edge set, $E$, whose elements are unordered pair of vertices. A $\mu$-orthonormal representation ($\mu$-OR) of $G$ is the assignment of a unit vector $\ket{v_i}\in \mathbb{C}^m$ for some integer, $m$, to each vertex $i\in V$ such that $|\langle v_i|v_j\rangle|\leq \mu$ whenever $\{i,j\}\in E$. Furthermore, with each $\mu$-OR, we associate a unit vector, $\ket{\psi}\in \mathbb{C}^m$, called the \textit{handle vector}. We define the $\mu$-approximate Lov\'asz theta function of a graph, $G$, as 
\begin{align}
    \vartheta_{\mu}(G):= \max \sum_{i\in V} |\langle \psi | v_i\rangle|^2,\label{eq:lovasz-theta}
\end{align}
where the maximization ranges over all $\mu$-ORs and  choices of handle vectors. Note that $\vartheta_0(G)\equiv \vartheta(G)$ is the standard Lov\'asz theta function of $G$~\cite{Lovasz_1979}. In fact, we have the following
\begin{lemma}$\lim_{\mu\to 0} \vartheta_\mu (G) = \vartheta(G)$.\label{lem:lovasz-approx}\end{lemma}
\begin{proof}
    As $\vartheta(G) \leq \vartheta_{\mu'}(G)\leq \vartheta_{\mu}(G)$ for $0\leq \mu'\leq \mu$, we have $\lim_{\mu\to 0} \vartheta_{\mu}(G) \geq \vartheta(G)$. To show the converse, choose a sequence, $\{\mu_m\}$ with $\lim_{m\to \infty} \mu_m=0$. Consider a corresponding sequence of optimal $\mu$-OR and handle vectors, equivalently, the sequence of their Gram matrices, $X_{\mu_m}$. As the $\mu$-OR and handle vectors are all unit vectors, $X_{\mu_m}$ are positive semidefinite (PSD) matrices with diagonal entries equal to one. As all such matrices form a compact set, there exists a subsequence $\mu_{m_l}\to 0$, such that $X_{\mu_{m_l}}$ converges to a PSD matrix $X^*$. Hence, the diagonal entries of $X^*$ are one and for all $\{i,j\}\in E$, $|X^*_{ij}| =\lim_{l\to \infty} |(X_{\mu_{m_l}})_{i,j}| = 0$. Therefore, $X^*$ is a feasible solution to the optimization for $\vartheta(G)$, i.e., $\lim_{\mu\to 0} \vartheta_{\mu}(G) \leq \vartheta(G)$.
\end{proof}
The Kneser graph, $KG_{n,k}$ consists of vertices representing the $k$-subsets of $[n]:=\{1,2,\dots,n\}$, with an edge between two vertices representing disjoint $k$-subsets. Theorem 13 of~\cite{Lovasz_1979} shows that when $n\geq 2k$,
\begin{align}\label{eq:lovasz-Kneser}
    \vartheta(KG_{n,k}) = {n-1 \choose k-1}.
\end{align}
We have derived the exact expression for the $\mu$-approximate Lov\'asz theta function of the graph $KG_{n,k}$ and included in the Supplemental Material.

 $k$-extendibility of quantum states and channels is a powerful tool in quantum information theory with several applications~\cite{Brand_o_2012,Li_2018,Kaur_2019,Khatri_2017}. A bipartite state $\rho_{AB}$ (and its isomorphic channel) is $k$-extendible for some integer $k\geq 1$ with respect to subsystem $B$ if there exists a quantum state $\rho_{AB_1B_2\dots B_k}$ such that $\rho_{AB} = \rho_{AB_j} = \Tr_{B^k\setminus B_j }[\rho_{AB_1B_2\dots B_k}]$. 

It is easy to see that $\rho^{\lambda}_{AB}$ with $\lambda =1/k$ is $k$-extendible. Indeed with probability $1/k$ the receiver obtains the EPR pair with $A$, while it goes to the environment with probability $(k-1)/k$. One can further split the environment into $(k-1)$ parties each receiving the EPR pair with probability $1/k$. Thus each of the $(k-1)$ parties in the environment is equivalent to the receiver $B$, demonstrating the $k$-extendibility.

\textit{Main Results ---}
Consider a set of $n$ systems, $\{B_i: i\in [n]\}$. For a $k$-subset of $[n]$, $S \subseteq_k [n]$, we denote $B_S$ as the set of subsystems indexed by the elements of $S$, i.e., $B_S := \bigotimes_{i\in S} B_i$. We introduce the following definition:
\begin{definition}[$(p,q)$-extendibility]\label{def:fract-ext}
    A state $\rho_{AB}$ is $(p,q)$-extendible with $p\geq q$ if there is a state $\sigma_{AB_1B_2\dots B_{p+q}}$, such that for every $S\subseteq_q [p+q]$, there exists a channel $\mathcal{R}_S:B_S \to B$ satisfying $\mathrm{id}_{A}\otimes \mathcal{R}_S(\sigma_{AB_S}) = \rho_{AB}$.
\end{definition}
In particular if a $(p,q)$-extension, $\sigma_{AB_{[p+q]}}:=\sigma_{AB_1B_2\dots B_{p+q}}$ exists that is permutation-invariant over $B_i$'s, $\rho_{AB}$ is called symmetrically $(p,q)$-extendible. Unless explicitly mentioned, we always assume that the state is extendible with respect to the second subsystem. Now we establish three important properties of $(p,q)$-extendibility.
\begin{lemma}
    The tensor product of two $(p,q)$-extendible states is $(p,q)$-extendible.
\end{lemma}
\begin{proof}
Let the states $\rho_{AB}$ and $\tau_{CD}$ be $(p,q)$-extendible on $B$ and $D$ respectively. Therefore, there exist states $\sigma_{AB_{[p+q]}}$ and $\omega_{CD_{[p+q]}}$ such that for every $S\subseteq_q [p+q]$, there exist channels $\mathcal{R}_S:B_S\to B$ and $\mathcal{T}_S:D_S\to D$ such that $\mathrm{id}_{A}\otimes \mathcal{R}_S(\sigma_{AB_S}) = \rho_{AB}$ and $\mathrm{id}_{C}\otimes \mathcal{T}_S(\omega_{CD_S}) = \tau_{CD}$. Now consider the state $\Xi_{AC B_{[p+q]}D_{[p+q]}}:= \sigma_{AB_{[p+q]}}\otimes \omega_{CD_{[p+q]}}$. For every $S\subseteq_q [p+q]$, we have that 
\begin{align*}
    &\mathrm{id}_{AC}\otimes (\mathcal{R}_S\otimes \mathcal{T}_S)(\Xi_{AC B_{S}D_{S}})\\ =\; &(\mathrm{id}_A\otimes \mathcal{R}_S)(\sigma_{AB_{S}}) \otimes (\mathrm{id}_C\otimes \mathcal{T}_S)(\omega_{CD_{S}}) = \rho_{AB}\otimes \tau_{CD}.
\end{align*}
Therefore, $\rho_{AB}\otimes \tau_{CD}$ is $(p,q)$-extendible on system $BD$ with the $(p,q)$-extension, $\Xi_{AC B_{[p+q]}D_{[p+q]}}$. 
\end{proof}

Applied iteratively, this implies: 

\begin{corollary}[Invariance under tensor power] \label{lemma:tensor_inv}
    If a state $\rho_{AB}$ is $(p,q)$-extendible, then so is $\rho^{\otimes n}_{AB}$, $\forall n\in \mathbb{N}$.
\end{corollary}
\begin{lemma}[Monotonicity under post-processing]\label{lem:post-process}
    If a state $\rho_{AB}$ is $(p,q)$-extendible, then so is $\mathrm{id}_A\otimes\mathcal{D}(\rho_{AB})$ for all channels, $\mathcal{D}:B\to B'$.
\end{lemma}
\begin{proof}
    Note that $\sigma_{AB_{[p+q]}}$ is a $(p,q)$-extension of $\mathrm{id}_A\otimes\mathcal{D}(\rho_{AB})$ with $\mathcal{D}\circ \mathcal{R}_S$ being the processing on subsystem $B_S$ for $S\subseteq_q[p+q]$.
\end{proof}
\begin{lemma}[Monotonicity under pre-processing]\label{lem:pre_process}
    If a state $\rho_{AB}$ is $(p,q)$-extendible, then so is $\mathcal{E}\otimes\mathrm{id}_B(\rho_{AB})$ for all CP maps, $\mathcal{E}:A\to R$ with $\Tr \mathcal{E}(\rho_{A}) = 1$.
\end{lemma}
\begin{proof}
        We note the fact that $\mathcal{E}\otimes \mathrm{id}_{B_{[p+q]}}(\sigma_{AB_{[p+q]}})$ is a $(p,q)$-extension of $\mathcal{E}\otimes\mathrm{id}_B(\rho_{AB})$ with $\mathcal{R}_S$ being the processing on subsystem $B_S$ for $S\subseteq_q[p+q]$. Obviously, $\Tr \mathcal{E}(\sigma_A) = \Tr\mathcal{E}(\rho_A)=1$.
\end{proof}
We now specialize to the Choi states of quantum erasure channels with transmission probability less than one half and study their fractional extendibility. 
\begin{lemma}[Lower bound on the extendibility]\label{lem:low_bound}
    The state $\rho^\lambda_{AB}$ with transmission probability $\lambda = q/(p+q)$, is $(p,q)$-extendible.
\end{lemma}
\begin{proof}
    It is easy to see that $\ket{\sigma}_{AB_{[p+q]}}= (p+q)^{-1/2} \sum_{i=1}^{p+q} \ket{\phi}^{AC_i}\ket{e}^{\otimes F_{[p+q]\setminus i}}$ is a (symmetric) $(p,q)$-extension of $\rho^\lambda_{AB}$, where $B_i \cong C_i \oplus F_i$. The processing to obtain $\rho^\lambda_{AB}$ involves detecting whether all of the $B_i$'s in a $q$-subset are in the flag state $\ket{e}$. If yes - which happens with probability $p/(q+p)$ - we raise an erasure flag. If not - which happens with probability $q/(p+q)$ - we obtain a maximally entangled state with $A$.
\end{proof}
\begin{figure}
    \centering
    \includegraphics[width=0.6\linewidth]{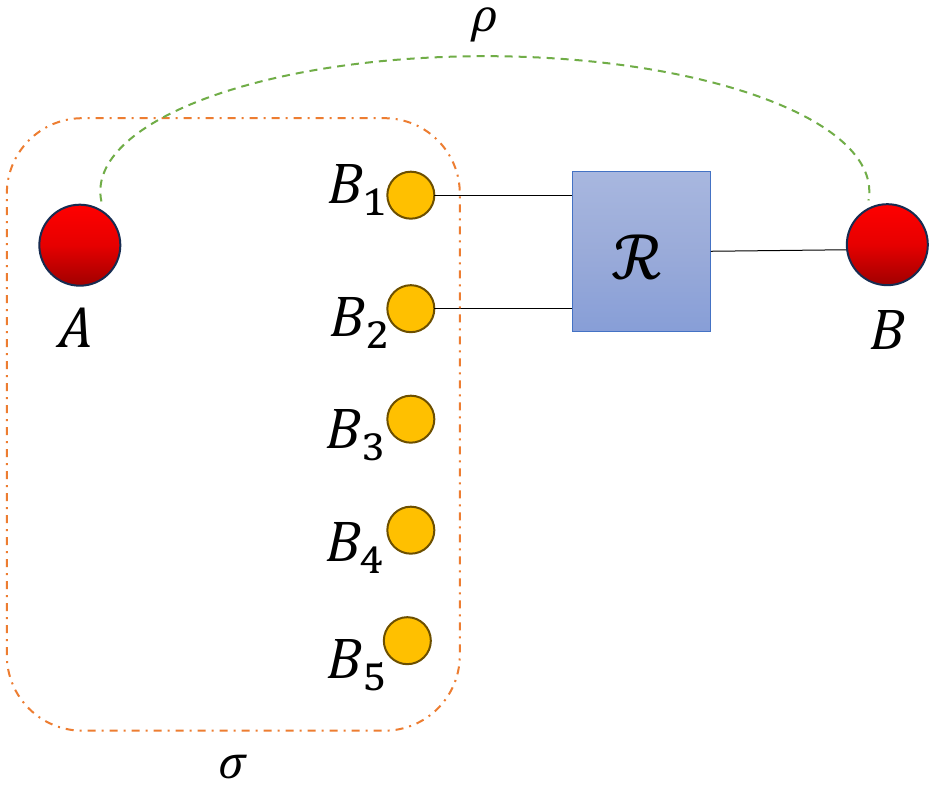}
    \caption{A $(3,2)$-extension, $\sigma_{AB_{[5]}}$. Any two $B_i$'s can be jointly processed to output the $(3,2)$-extendible state, $\rho_{AB}$.}
    \label{fig:3,2-ext}
\end{figure}
\begin{lemma}[Robust upper bound on the extendibility]\label{lem:upp_bound}
    Consider $\rho^\lambda_{AB}$ with transmission probability $\lambda = q/(p+q)$. If $p'/q'>p/q$, then $\rho^\lambda_{AB}$ is not in the closure of the set of all $(p',q')$-extendible states, denoted by $\overline{\mathrm{Ext}_{(p',q')}}$.
\end{lemma}
Lemma \ref{lem:upp_bound} is a direct result of the following proposition.  Informally, it states: for any multipartite quantum state on $AB_1B_2\dots B_n$, the maximum average probability with which processing on any $k$ subsystems among $\{B_i\}$ outputs a maximally entangled state with $A$ is the fraction $k/n$ whenever $n\geq 2k$. This is exactly the fraction one obtains by giving half of the maximally entangled state to each of the $n$ parties with equal probability. Hence, error correction is of no benefit in this regime. 
\begin{prop}[Distributed monogamy of entanglement]\label{prop:ent-frac}
    Fix a state $\sigma_{AB_1B_2\dots B_n}$, an integer, $k \leq n/2$, and a collection $\{\lambda_S\in [0,1]\}_{S\subseteq_k [n]}$. If, for some $\epsilon \geq 0$, there exists a set of quantum channels, $\{\mathcal{R}_S:B_S\to D_S\cong C\oplus F\}_{S\subseteq_k [n]}$
    such that for every $S\subseteq_k [n]$,
    $\frac{1}{2}\norm{\mathrm{id}_A\otimes \mathcal{R}_S(\sigma_{AB_S}) - \rho^{\lambda_S}_{AD_S}}_1 \leq \epsilon$, then ${n \choose k}^{-1}\sum_{S\subseteq_k [n]} \lambda_S \leq k/n + \Delta_\epsilon$, where $\lim_{\epsilon\to 0} \Delta_\epsilon = 0$. 
\end{prop}
\begin{proof}
We first present the argument in the special case needed below, where $\lambda_S>\epsilon$  for all $S$.
The general case of arbitrary $\lambda_S\in [0,1]$ is obtained by a threshold argument which is given in the supplemental material.

    We use the purification $\ket{\sigma}_{AB_{[n]} E}$. For all $k$-subsets, $S\subseteq_k [n]$, denote the unitary dilation of the channel, $\mathcal{R}_S:B_S\to D_S$ by $U_S: B_SE_S\to D_S E_S$. Here we have introduced ancillary systems of sufficiently large dimension, $E_i$, prepared in a pure state, denoted by $\ket{0}_{E_i}$, associated with each $B_i$. Therefore, the composite pure state of all the components can be written as
    $\ket{\tilde{\sigma}}_{AB_{[n]}E_{[n]}E} := \ket{\sigma}_{AB_{[n]}E}\ket{0}_{E_1}\ket{0}_{E_2}\dots \ket{0}_{E_n}$. After we process the $k$-subset, $S$, the pure state is given by $U_S\ket{\tilde{\sigma}}$. Here we implicitly assume a tensor product with the identity operator on the rest of the systems. Clearly, $\omega_{AD_S}:= \mathrm{id}_A\otimes \mathcal{R}_S(\sigma_{AB_S}) = \Tr_{E_{[n]}B_{S^c}E}{(U_S\ketbra{\tilde{\sigma}}{\tilde{\sigma}}U^\dagger_S)}$.

    We define the projector $ \Pi^S_C$ on $C$, the $d$-dimensional subspace of $D_S$ that is isomorphic to $A$. The probability with which $\omega_{AD_S}$ is detected in the subspace $AC$ is given by  $p_S = \Tr\left\{\Pi^S_C\omega_{AD_S}\right\} = \bra{\tilde{\sigma}} M_S \ket{\tilde{\sigma}}$, where $M_S = U_S^\dagger (\Pi^S_C\otimes \mathbb{I}_{E_S}) U_S$. As $U_S$ is a unitary and $\Pi^S_C$ is a projector, we conclude that $M_S$ is a projector: $M^\dagger_S = M_S$ and $M_S^2 = M_S$. Therefore, given the state $\ket{\tilde{\sigma}}$,Thus $M_S$ represents the non-erasure outcome of the
processing associated with $S$. Let the post-measurement normalised state be denoted by $\ket{v_S}:= \frac{M_S\ket{\tilde{\sigma}}}{\sqrt{\bra{\tilde{\sigma}} M_S \ket{\tilde{\sigma}}}} = \frac{M_S\ket{\tilde{\sigma}}}{\sqrt{p_S}}$.

    Furthermore, for any two disjoint $k$-subsets, $S$ and $T$, the joint probability that both $S$ and $T$ measure their processed $D$ systems in the respective $C$ subspaces is given by
    \begin{align}
        r_{ST} &= \Tr{\left(\mathbb{I}_A\otimes\Pi^S_C\otimes \Pi_C^T\right) U_SU_T \ketbra{\tilde{\sigma}}{\tilde{\sigma}}U_S^\dagger U_T^\dagger} \nonumber\\
        &= \bra{\tilde{\sigma}}U^\dagger_S \Pi^S_C U_S\otimes U^\dagger_T \Pi^T_C U_T \ket{\tilde{\sigma}}\nonumber\\
        &= \bra{\tilde{\sigma}}M_S M_T\ket{\tilde{\sigma}} = \sqrt{p_Sp_T} \langle v_T|v_S\rangle. \label{eq:overlap}
    \end{align}
    When $\epsilon = 0$, $r_{ST}$ is the probability with which both systems $D_S$ and $D_T$ are maximally entangled with $A$, hence $r_{ST}$ must vanish due to the monogamy of entanglement. More generally, we prove that for any two disjoint subsets $S$ and $T$, the overlap $\langle v_T|v_S\rangle$ is bounded as follows 
\begin{align}\label{lem:overlap}
    |\langle v_T|v_S\rangle| &\leq \frac{2\epsilon}{\left[(1-1/d)\sqrt{(\lambda_S-\epsilon)(\lambda_T-\epsilon)}\right]}=: f_{S,T}(\epsilon).
\end{align}
Since the collection \(\{\lambda_S\}_{S\subseteq_k[n]}\) is
fixed and finite, and $\lambda_S>0$ for every $S$,
$f(\epsilon):=\max_{S\cap T=\varnothing}f_{S,T}(\epsilon)
\to 0$ as $\epsilon\to0$. The proof is included in the End Matter.

    Now the average probability with which a randomly chosen $k$-subset, $S$, processed with $\mathcal{R}_S$, outputs a state in the non-erasure subspace is given by 
    \begin{align}
        \mathcal{P}_{ave} &=  {n\choose k}^{-1}\sum_{S\subseteq_k [n]} p_S\geq \left[{n\choose k}^{-1}\sum_{S\subseteq_k [n]} \lambda_S\right] - \epsilon.
    \end{align} 
    Here we have used Eq.~\eqref{eq:prob}. Furthermore, the probabilities are given by $p_S = |\langle \tilde{\sigma}|v_S\rangle|^2$. Therefore, we can write
    \begin{align}
        \mathcal{P}_{ave} = \frac{1}{{n\choose k}} \left(\sum_{S\subseteq_k [n]} |\langle \tilde{\sigma}|v_S\rangle|^2\right).
    \end{align}
    Now we note that $\{\ket{v_S}: S\subseteq_k [n]\}$ forms an $f(\epsilon)$-orthonormal representation of the Kneser graph, $KG_{n,k}$, where $f(\epsilon)=\max_{S,T}\{f_{S,T}(\epsilon)\}$. In fact, the term in the parentheses above is the objective function for the maximization problem in Eq.~\eqref{eq:lovasz-theta}, corresponding to a feasible solution. Hence, we obtain  $\mathcal{P}_{ave} \leq {n \choose k}^{-1}\vartheta_{f(\epsilon)}(KG_{n,k}) =k/n + {n \choose k}^{-1}\left[\vartheta_{f(\epsilon)}(KG_{n,k}) -\vartheta(KG_{n,k})\right]$, where we have used Eq.~\eqref{eq:lovasz-Kneser}. Combining everything, we finally get
    \begin{align}\label{eq:lam-av}
        \frac{\sum_{S\subseteq_k [n]} \lambda_S }{{n \choose k}}\leq \frac{k}{n} + \epsilon + \frac{\left[\vartheta_{f(\epsilon)}(KG_{n,k}) -\vartheta(KG_{n,k})\right]}{{n \choose k}}. 
    \end{align}
    By Lemma~\ref{lem:lovasz-approx}, we argue that $\Delta_\epsilon:= \epsilon + {n \choose k}^{-1}\left[\vartheta_{f(\epsilon)}(KG_{n,k}) -\vartheta(KG_{n,k})\right]$ approaches $0$ as $\epsilon\to 0$.
\end{proof}
\textit{Proof of Lemma~\ref{lem:upp_bound}:} Suppose that for $\lambda = q/{(p+q)}$, $\rho^\lambda_{AB} \in \overline{\mathrm{Ext}_{(p',q')}}$ with $p'/q' > p/q$. Hence, for all $\epsilon>0$ there exists a $(p',q')$-extendible state, $\omega_{AB}$ such that $\frac{1}{2}\norm{\omega_{AB} - \rho^\lambda_{AB}}_1 < \epsilon$. From Definition~\ref{def:fract-ext} and Proposition~\ref{prop:ent-frac}, we conclude that $\lambda \leq q'/{(p'+q')} + \mu$ for all $\mu>0$. This is in contradiction to the fact that $\lambda = q/{(p+q)} > q'/{(p'+q')}$. \qed

Now we have acquired all the necessary components to prove Hastings' conjecture on channel simulation ~\cite{Hastings2008}.
\begin{theorem}[Hastings' conjecture]\label{thm:hastings}
    $\mathcal{N}_{\gamma}$ cannot simulate $\mathcal{N}_{\lambda}$ where $\gamma < \lambda \leq \frac{1}{2}$.
\end{theorem}
\begin{proof}
We first prove the case for rational numbers $\gamma = q/(p+q)$ and $\lambda = s/(r+s)$. As $\gamma < \lambda \leq \frac{1}{2}$, we have that $p/q>r/s\geq 1$. Therefore, using Lemma~\ref{lem:upp_bound}, we conclude that $\rho^{\lambda}_{AB}\notin \overline{\mathrm{Ext}_{(p,q)}}$.

The Choi state of $n$ copies of an erasure channel, $\mathcal{N}_\gamma$, is given by $\tilde{\rho}_{A^nB^n} := (\rho^\gamma_{AB})^{\otimes n} = \mathcal{N}_{\gamma}^{\otimes n}(\ketbra{\phi}{\phi}^{\otimes n}_{AA'}) \equiv \mathcal{N}_{\gamma}^{\otimes n}(\ketbra{\phi}{\phi}_{A^nA'^n})$, where $A^n= A_{[n]}$, $B^n=B_{[n]}$ and $A'^n = A'_{[n]}$. Additionally, for any encoder $\mathcal{E}_n: R\to A^n$, denote by $\mathcal{E}_n^T$ a CP unital map with Kraus operators being the transpose of the Kraus operators of $\mathcal{E}_n$ in some chosen decomposition. 

As $\rho^{\gamma}_{AB}$ is $(p,q)$-extendible (Lemma~\ref{lem:low_bound}), so is $\tilde{\rho}_{A^nB^n}$ (Corollary~\ref{lemma:tensor_inv}). As $\tilde{\rho}_{A^nB^n}$ is $(p,q)$-extendible, so is $\mathrm{id}_{A^n} \otimes \mathcal{D}_n(\tilde{\rho}_{A^nB^n})$ for any decoder, $\mathcal{D}_n: B^n\to B$ (Lemma~\ref{lem:post-process}). As $\mathrm{id}_{A^n} \otimes \mathcal{D}_n(\tilde{\rho}_{A^nB^n})$ is $(p,q)$-extendible, so is $ \omega^{(n)}_{RB} := \frac{(\dim A)^n}{\dim R}\mathcal{E}_n^T \otimes \mathcal{D}_n(\tilde{\rho}_{A^nB^n})$, for any encoder $\mathcal{E}_n: R\to A^n$ (Lemma~\ref{lem:pre_process}).

Using the transpose trick~\footnote{For a $d_1\times d_2$ matrix $A$, we have $(A\otimes \mathbb{I}_{d_2}) \ket{\beta_{d_2}} = (\mathbb{I}_{d_1}\otimes A^T)\ket{\beta_{d_1}}$. Here $\ket{\beta_n} = \sum_{i=1}^n \ket{ii}$ is the unnormalized maximally entangled state on $\mathbb{C}^n\otimes\mathbb{C}^n$. The factor of $\frac{(\dim A)^n}{\dim R}$ in the proof takes care of the normalization.}, we have $\omega^{(n)}_{RB} = \mathrm{id}_R\otimes (\mathcal{D}_n\circ\mathcal{N}_{\gamma}^{\otimes n}\circ\mathcal{E}_n)(\ketbra{\phi}{\phi}_{RR'}) = \mathcal{J}_{\mathcal{D}_n\circ\mathcal{N}_{\gamma}^{\otimes n}\circ\mathcal{E}_n}$. Recall from the definition of channel simulation that if $\mathcal{N}_{\gamma}$ could simulate $\mathcal{N}_{\lambda}$, we would have a sequence of $(p,q)$-extendible states, $\{\omega^{(n)}_{RB}\}$ converging to $\rho^{\lambda}_{AB}$ which is not in $\overline{\mathrm{Ext}_{(p,q)}}$, arriving at a contradiction. Thus the claim is proven for rational $\gamma$ and $\lambda$.

For irrational $\gamma$ and $\lambda$, note that in between any two irrational numbers, there exist two rational numbers. Hence, such a simulation could be used to contradict what we have just proven above. This finishes the proof.
\end{proof}

Additionally, we have calculated a strictly positive lower bound for the asymptotic simulation error, i.e., the trace distance between the simulated channel and the target channel. We present the dimension-independent bound here in the main text. The dimension-dependent bound and its derivation can be found in the supplemental material.
\begin{theorem}[Quantitative Hastings' Conjecture]
    Let $\gamma < \lambda \leq \frac{1}{2}$. For every sequence of encoder--decoder pairs, $\{(\mathcal{E}_n,\mathcal{D}_n)\}$, we have
    \begin{align*}
       \liminf_{n\to \infty} &\frac{1}{2}\norm{\mathcal{J}_{\mathcal{N}_\lambda}- \mathcal{J}_{\mathcal{D}_n\circ\mathcal{N}_{\gamma}^{\otimes n}\circ\mathcal{E}_n}}_1\nonumber\\
       &\geq \lambda - \frac{1}{2}\left(\sqrt{
            \left(4-5\gamma\right)^2+ 16\lambda(1-\gamma)} -  \left(4-5\gamma\right) \right)
    \end{align*}
\end{theorem}

\textit{Discussions ---} We have introduced a notion of fractional extendibility, proved several of its useful properties, and established the distributed monogamy of entanglement (DME), a fundamental limit to probabilistic extraction of an EPR pair from a subset of quantum systems. In particular, we focus on the quantum erasure channel and provide a simple relation between its fractional extendibility and its erasure probability using DME. Finally we prove that a more-than-$50\%$-erasure channel cannot simulate a better erasure channel.

It is easy to extend the fractional extendibility arguments, developed for the erasure channels in this paper, to the qubit amplitude damping (AD) channel~\cite{Nielsen_Chuang_2010,Leung_1997}. For instance, a qubit AD channel with damping parameter $p/(q+p)$ has a symmetric $(p,q)$-extension. As in the proof of Lemma~\ref{lem:low_bound}, this can be constructed by creating an EPR pair on two qubits, $AB$ and using an appropriate isometry, $U:B\to B_{[p+q]}$ to \textit{distribute} the excitation among $B_i$'s with equal probability. To upper bound the fractional extendibility, one can use the following line of reasoning. With the dual-rail encoding, a pair of qubit amplitude damping channels with damping parameter $p/(q+p)$ can simulate one qubit erasure channel with erasure probability $p/(p+q)$~\cite{Duan_2010,Kubica_2023}. If the AD channel were $(r,s)$-extendible with $r/s>p/q$, then using this simulation, one could violate DME. 

One can similarly generalize Theorem~\ref{thm:hastings} to account for qubit AD channels. Indeed, one obtains a more general version of Hastings' conjecture: for any real number $x\geq 1/2$, arbitrary number of erasure and AD channels with noise parameters (i.e., erasure probability and damping respectively) greater than $x$ cannot simulate a single erasure or AD channel with noise parameter less than $x$.

\textit{Note added ---} Shortly after the first version of this work was posted, we became aware of an independent proof~\cite{alhejji2026constraintsrecoveringquantuminformation} by Alhejji et al. of the erasure-simulation conjecture. 

\textit{Acknowledgements ---} We thank Debbie Leung for helpful discussions and comments. GS is supported under NSERC-NSF alliance grant ALLRP-586858-2023 and NSERC Discovery grant RGPIN-2025-02094. RGA acknowledges the support of the Institute for Quantum Computing and Mike and Ophelia Lazaridis Graduate Fellowship.


\bibliography{apssamp}

@ARTICLE{Horodecki_2008,
  author={Horodecki, Karol and Pankowski, Lukasz and Horodecki, Michal and Horodecki, Pawel},
  journal={IEEE Transactions on Information Theory}, 
  title={Low-Dimensional Bound Entanglement With One-Way Distillable Cryptographic Key}, 
  year={2008},
  volume={54},
  number={6},
  pages={2621-2625},
  doi={10.1109/TIT.2008.921709}
}

@ARTICLE{Smith_2008,
  author={Smith, Graeme and Smolin, John A. and Winter, Andreas},
  journal={IEEE Transactions on Information Theory}, 
  title={The Quantum Capacity With Symmetric Side Channels}, 
  year={2008},
  volume={54},
  number={9},
  pages={4208-4217},
  doi={10.1109/TIT.2008.928269}
}

@article{Simth-Yard_2008,
    author = {Graeme Smith  and Jon Yard },
    title = {Quantum Communication with Zero-Capacity Channels},
    journal = {Science},
    volume = {321},
    number = {5897},
    pages = {1812-1815},
    year = {2008},
    doi = {10.1126/science.1162242}
}

@article{Bennett_1997,
  title = {Capacities of Quantum Erasure Channels},
  author = {Bennett, Charles H. and DiVincenzo, David P. and Smolin, John A.},
  journal = {Phys. Rev. Lett.},
  volume = {78},
  issue = {16},
  pages = {3217--3220},
  numpages = {0},
  year = {1997},
  month = {Apr},
  publisher = {American Physical Society},
  doi = {10.1103/PhysRevLett.78.3217},
  url = {https://link.aps.org/doi/10.1103/PhysRevLett.78.3217}
}

@phdthesis{Alhejji2023,
author={Alhejji,Mohammad A.},
year={2023},
title={Some Problems Concerning Quantum Channels and Entropies},
journal={ProQuest Dissertations and Theses},
pages={93},
isbn={9798379527198},
url={https://proxy.lib.uwaterloo.ca/login?url=https://www.proquest.com/dissertations-theses/some-problems-concerning-quantum-channels/docview/2813838508/se-2}
}

@misc{ahmed2026,
      title={No-Go Theorem for Gaussian Quantum Repeaters from Fractional Extendibility}, 
      author={Rabsan Galib Ahmed and Graeme Smith},
      year={2026},
      eprint={2606.05097},
      archivePrefix={arXiv},
      primaryClass={quant-ph},
      url={https://arxiv.org/abs/2606.05097}, 
}

@ARTICLE{Wootters1982-cz,
  title     = "A single quantum cannot be cloned",
  author    = "Wootters, W K and Zurek, W H",
  journal   = "Nature",
  publisher = "Springer Science and Business Media LLC",
  volume    =  299,
  number    =  5886,
  pages     = "802--803",
  month     =  oct,
  year      =  1982,
  url       =  "https://www.nature.com/articles/299802a0"
}

@article{CSW2014,
  title = {Graph-Theoretic Approach to Quantum Correlations},
  author = {Cabello, Ad\'an and Severini, Simone and Winter, Andreas},
  journal = {Phys. Rev. Lett.},
  volume = {112},
  issue = {4},
  pages = {040401},
  numpages = {5},
  year = {2014},
  month = {Jan},
  publisher = {American Physical Society},
  doi = {10.1103/PhysRevLett.112.040401},
  url = {https://link.aps.org/doi/10.1103/PhysRevLett.112.040401}
}

@ARTICLE{Lovasz_1979,
  author={Lovasz, L.},
  journal={IEEE Transactions on Information Theory}, 
  title={On the Shannon capacity of a graph}, 
  year={1979},
  volume={25},
  number={1},
  pages={1-7},
  doi={10.1109/TIT.1979.1055985}
}

@article{Bennett_2014,
   title={The Quantum Reverse Shannon Theorem and Resource Tradeoffs for Simulating Quantum Channels},
   volume={60},
   ISSN={1557-9654},
   url={http://dx.doi.org/10.1109/TIT.2014.2309968},
   DOI={10.1109/tit.2014.2309968},
   number={5},
   journal={IEEE Transactions on Information Theory},
   publisher={Institute of Electrical and Electronics Engineers (IEEE)},
   author={Bennett, Charles H. and Devetak, Igor and Harrow, Aram W. and Shor, Peter W. and Winter, Andreas},
   year={2014},
   month=May, pages={2926–2959} 
}

@article{Brand_o_2012,
  title = {Detection of Multiparticle Entanglement: Quantifying the Search for Symmetric Extensions},
  author = {Brand\~ao, Fernando G. S. L. and Christandl, Matthias},
  journal = {Phys. Rev. Lett.},
  volume = {109},
  issue = {16},
  pages = {160502},
  numpages = {5},
  year = {2012},
  month = {Oct},
  publisher = {American Physical Society},
  doi = {10.1103/PhysRevLett.109.160502},
  url = {https://link.aps.org/doi/10.1103/PhysRevLett.109.160502}
}

@article{Li_2018,
   title={Squashed Entanglement, $k$-Extendibility, Quantum Markov Chains, and Recovery Maps},
   volume={48},
   ISSN={1572-9516},
   url={http://dx.doi.org/10.1007/s10701-018-0143-6},
   DOI={10.1007/s10701-018-0143-6},
   number={8},
   journal={Foundations of Physics},
   publisher={Springer Science and Business Media LLC},
   author={Li, Ke and Winter, Andreas},
   year={2018},
   month=Feb, pages={910–924} }

@article{Khatri_2017,
   title={Numerical evidence for bound secrecy from two-way postprocessing in quantum key distribution},
   volume={95},
   ISSN={2469-9934},
   url={http://dx.doi.org/10.1103/PhysRevA.95.042320},
   DOI={10.1103/physreva.95.042320},
   number={4},
   journal={Physical Review A},
   publisher={American Physical Society (APS)},
   author={Khatri, Sumeet and Lütkenhaus, Norbert},
   year={2017},
   month=Apr 
}

@article{Kaur_2019,
  title = {Extendibility Limits the Performance of Quantum Processors},
  author = {Kaur, Eneet and Das, Siddhartha and Wilde, Mark M. and Winter, Andreas},
  journal = {Phys. Rev. Lett.},
  volume = {123},
  issue = {7},
  pages = {070502},
  numpages = {7},
  year = {2019},
  month = {Aug},
  publisher = {American Physical Society},
  doi = {10.1103/PhysRevLett.123.070502},
  url = {https://link.aps.org/doi/10.1103/PhysRevLett.123.070502}
}

@article{Kunjwal_2019,
   title={Beyond the Cabello-Severini-Winter framework: Making sense of contextuality without sharpness of measurements},
   volume={3},
   ISSN={2521-327X},
   url={http://dx.doi.org/10.22331/q-2019-09-09-184},
   DOI={10.22331/q-2019-09-09-184},
   journal={Quantum},
   publisher={Verein zur Forderung des Open Access Publizierens in den Quantenwissenschaften},
   author={Kunjwal, Ravi},
   year={2019},
   month=Sept, pages={184} 
}

@ARTICLE{DSW_2013,
  author={Duan, Runyao and Severini, Simone and Winter, Andreas},
  journal={IEEE Transactions on Information Theory}, 
  title={Zero-Error Communication via Quantum Channels, Noncommutative Graphs, and a Quantum Lovász Number}, 
  year={2013},
  volume={59},
  number={2},
  pages={1164-1174},
  doi={10.1109/TIT.2012.2221677}}

@misc{Hastings2008,
      title={Notes on Some Questions in Mathematical Physics and Quantum Information}, 
      author={M. B. Hastings},
      year={2014},
      eprint={1404.4327},
      archivePrefix={arXiv},
      primaryClass={quant-ph},
      url={https://arxiv.org/abs/1404.4327}, 
}

@article{Berta_2011,
   title={The Quantum Reverse Shannon Theorem Based on One-Shot Information Theory},
   volume={306},
   ISSN={1432-0916},
   url={http://dx.doi.org/10.1007/s00220-011-1309-7},
   DOI={10.1007/s00220-011-1309-7},
   number={3},
   journal={Communications in Mathematical Physics},
   publisher={Springer Science and Business Media LLC},
   author={Berta, Mario and Christandl, Matthias and Renner, Renato},
   year={2011},
   month=Aug, pages={579–615} 
}

@ARTICLE{Fang_2020,
  author={Fang, Kun and Wang, Xin and Tomamichel, Marco and Berta, Mario},
  journal={IEEE Transactions on Information Theory}, 
  title={Quantum Channel Simulation and the Channel’s Smooth Max-Information}, 
  year={2020},
  volume={66},
  number={4},
  pages={2129-2140},
  doi={10.1109/TIT.2019.2943858}
}

@ARTICLE{Duan_2016,
  author={Duan, Runyao and Winter, Andreas},
  journal={IEEE Transactions on Information Theory}, 
  title={No-Signalling-Assisted Zero-Error Capacity of Quantum Channels and an Information Theoretic Interpretation of the Lovász Number}, 
  year={2016},
  volume={62},
  number={2},
  pages={891-914},
  doi={10.1109/TIT.2015.2507979}
}

@INPROCEEDINGS{Duan_2010,
  author={Duan, Runyao and Grassl, Markus and Ji, Zhengfeng and Zeng, Bei},
  booktitle={2010 IEEE International Symposium on Information Theory}, 
  title={Multi-error-correcting amplitude damping codes}, 
  year={2010},
  volume={},
  number={},
  pages={2672-2676},
  doi={10.1109/ISIT.2010.5513648}
}

@article{Kubica_2023,
  title = {Erasure Qubits: Overcoming the ${T}_{1}$ Limit in Superconducting Circuits},
  author = {Kubica, Aleksander and Haim, Arbel and Vaknin, Yotam and Levine, Harry and Brand\~ao, Fernando and Retzker, Alex},
  journal = {Phys. Rev. X},
  volume = {13},
  issue = {4},
  pages = {041022},
  numpages = {20},
  year = {2023},
  month = {Nov},
  publisher = {American Physical Society},
  doi = {10.1103/PhysRevX.13.041022},
  url = {https://link.aps.org/doi/10.1103/PhysRevX.13.041022}
}

@book{Nielsen_Chuang_2010, 
place={Cambridge}, 
title={Quantum Computation and Quantum Information: 10th Anniversary Edition}, 
publisher={Cambridge University Press}, 
author={Nielsen, Michael A. and Chuang, Isaac L.}, 
year={2010}}

@article{Leung_1997,
  title = {Approximate quantum error correction can lead to better codes},
  author = {Leung, Debbie W. and Nielsen, M. A. and Chuang, Isaac L. and Yamamoto, Yoshihisa},
  journal = {Phys. Rev. A},
  volume = {56},
  issue = {4},
  pages = {2567--2573},
  numpages = {0},
  year = {1997},
  month = {Oct},
  publisher = {American Physical Society},
  doi = {10.1103/PhysRevA.56.2567},
  url = {https://link.aps.org/doi/10.1103/PhysRevA.56.2567}
}

@ARTICLE{Duncan1976-cl,
  title     = "Norm inequalities for C*-algebras",
  author    = "Duncan, J and Taylor, P J",
  journal   = "Proc. R. Soc. Edinb. Sect. A",
  publisher = "Cambridge University Press (CUP)",
  volume    =  75,
  number    =  2,
  pages     = "119--129",
  year      =  1976
}

@ARTICLE{Conde2024-ht,
  title     = "Norm of the sum of two orthogonal projections",
  author    = "Conde, Cristian",
  journal   = "Banach J. Math. Anal.",
  publisher = "Springer Science and Business Media LLC",
  volume    =  18,
  number    =  3,
  month     =  jul,
  year      =  2024,
  copyright = "https://www.springernature.com/gp/researchers/text-and-data-mining"
}

@misc{alhejji2026constraintsrecoveringquantuminformation,
      title={Constraints on recovering quantum information after erasure}, 
      author={Mohammad A. Alhejji and Noah Lordi and Omkar Baraskar and Ariel Shlosberg and Emanuel Knill},
      year={2026},
      eprint={2607.17319},
      archivePrefix={arXiv},
      primaryClass={quant-ph},
      url={https://arxiv.org/abs/2607.17319}, 
}

@BOOK{Brouwer2011-vh,
  title     = "Spectra of graphs",
  author    = "Brouwer, Andries E and Haemers, Willem H",
  publisher = "Springer",
  series    = "Universitext",
  month     =  dec,
  year      =  2011,
  address   = "New York, NY",
  copyright = "https://www.springernature.com/gp/researchers/text-and-data-mining"
}

@BOOK{Godsil2001-xn,
  title     = "Algebraic Graph Theory",
  author    = "Godsil, Chris and Royle, Gordon",
  publisher = "Springer",
  series    = "Graduate Texts in Mathematics",
  month     =  apr,
  year      =  2001,
  address   = "New York, NY"
}
\begingroup
\appendix
\section{The overlap bound}
Here we provide the proof for Eq.~\eqref{lem:overlap}. For some $S \subseteq_k[n]$, consider the operator $\mathbb{I}_A\otimes\Pi^S_C - \Phi^S_{AC}$, where $\Phi^S_{AC}$ is the maximally entangled state on $AC \subset AD_S$. This is a positive semidefinite operator with eigenvalues no bigger than $1$. Therefore, from the variational characterization of trace distance, we obtain that 
\begin{align*}
    &\Tr{\left(\mathbb{I}_A\otimes\Pi^S_C - \Phi^S_{AC}\right) (\omega_{AD_S}-\rho^{\lambda_S}_{AD_S})} \leq\epsilon\\
    \implies &\Tr{\left(\mathbb{I}_A\otimes\Pi^S_C - \Phi^S_{AC}\right) \omega_{AD_S}} \leq\epsilon.
\end{align*}
 If the final joint state on $AD_SD_T$, together with the processing on the subset $T$, is given by $\omega_{AD_SD_T}$, then we can write,
 \begin{align*}
     &\Tr{\left(\mathbb{I}_A\otimes\Pi^S_C - \Phi^S_{AC}\right)\otimes \Pi^T_C \; \omega_{AD_SD_T}}\\
     \leq &\Tr{\left(\mathbb{I}_A\otimes\Pi^S_C - \Phi^S_{AC}\right) \omega_{AD_S}} \leq \epsilon
 \end{align*} We can swap the roles of $S$ and $T$ in the previous argument, to similarly obtain, $$\Tr{\left(\mathbb{I}_A\otimes\Pi^T_C - \Phi^T_{AC}\right)\otimes \Pi^S_C \; \omega_{AD_TD_S}} \leq \epsilon.$$ 
 Denoting $P_S = \Phi^S_{AC}\otimes \Pi^T_C$ and $P_T = \Phi^T_{AC}\otimes \Pi^S_C$ and adding the last two inequalities we obtain that $$2r_{ST} - \Tr{(P_S + P_T) \omega_{AD_SD_T}} \leq 2\epsilon$$. Rearranging, we have, 
    \begin{align*}
        2r_{ST} &\leq 2\epsilon + \Tr{(P_S + P_T)\; \omega_{AD_SD_T}} \nonumber\\
        & \leq 2\epsilon + \norm{P_S+P_T} \;r_{ST}.
    \end{align*}
    Here we use that for PSD operators, $X$ and $Y$, one has $$\Tr{XY} \leq \norm{X}\Tr {\Pi_{\mathrm{supp}(X)} Y},$$where $\norm{X}$ is the spectral norm of $X$ and $\Pi_{\mathrm{supp}(X)}$ is the projector onto the support of $X$. For two projectors, $P$ and $Q$, the following holds~\cite{Duncan1976-cl,Conde2024-ht}: 
    $$\norm{P+Q} = 1 + \norm{PQ}.$$
    For these specific projectors $P_S$ and $P_T$, the following crucial identity holds: $P_SP_TP_S = P_S/d^2$, implying $\norm{P_SP_T} = 1/d$, where $d = \dim A$. Therefore, combining everything, we obtain that $r_{ST} \leq 2\epsilon/(1-1/d)$.

    We can again use the variational characterisation of the trace distance to lower bound $p_S, p_T$: 
    \begin{align}\label{eq:prob}
        \Tr{\Pi_C^S(\rho_{AD_S}^{\lambda_S} - \omega_{AD_S})} \leq \epsilon \implies p_S \geq\lambda_S-\epsilon.
    \end{align}
Finally, using Eq.~\eqref{eq:overlap}, one can easily verify Eq.~\eqref{lem:overlap}.
\vspace{2mm}
\section{Generalization of Hastings' conjecture}
\begin{corollary}\label{cor:gen-hastings}
    For any real number $x\geq 1/2$, arbitrary uses of erasure and AD channels with noise parameters (i.e., erasure probability and damping respectively) greater than $x$ cannot simulate a single erasure or AD channel with noise parameter less than $x$.
\end{corollary}
\textit{Proof:}
Choose rational numbers \(z,y\) such that $x_0<z<y<x$.
Every available resource channel is a post-processing of
the corresponding erasure or AD channel
with noise parameter \(y\). Hence the simulated Choi
state is \((p,q)\)-extendible, where
\(y=p/(p+q)\).

If the target is an erasure channel with noise parameter
\(x_0\), additional erasure produces the erasure channel
with rational noise parameter \(z\). If the simulation was possible, it would place the
latter in the closure of the \((p,q)\)-extendible states,
contradicting Lemma~\ref{lem:upp_bound}.

If the target is an amplitude-damping channel, two copies
followed by the dual-rail construction produce an erasure
channel with noise parameter \(x_0\). Additional erasure
again produces the rational erasure channel with noise
parameter \(z\), yielding the same contradiction.\qed  
\endgroup

\clearpage
\onecolumngrid
\setcounter{section}{0}
\setcounter{subsection}{0}
\setcounter{equation}{0}
\setcounter{figure}{0}
\setcounter{table}{0}

\renewcommand{\thesection}{S\arabic{section}}
\renewcommand{\thesubsection}{S\arabic{section}.\arabic{subsection}}
\renewcommand{\theequation}{S\arabic{section}.\arabic{equation}}
\renewcommand{\thefigure}{S\arabic{figure}}
\renewcommand{\thetable}{S\arabic{table}}

\begin{center}
    {\large\bfseries Supplemental Material for}\\[4pt]
    {\large \bfseries ``Distributed Monogamy of Entanglement limits Quantum Channel Simulation''}\\[10pt]
{\normalsize
 Rabsan Galib Ahmed$^{1,2}$ and Graeme Smith$^{1,2}$
}\\[7pt]

{\small
$^{1}$Department of Applied Mathematics, University of Waterloo, Ontario N2L 3G1, Canada.
}\\[3pt]

{\small
$^{2}$Institute for Quantum Computing, University of Waterloo, Ontario N2L 3G1, Canada.
}\\[3pt]

\end{center}

In this Supplemental Material, we provide a proof of the lower bound on the erasure-channel simulation error, thereby quantifying the version of Hastings' conjecture proved in the main text. We derive explicit bounds on the $\mu$-approximate Lov\'asz theta function and complete the proof of Proposition~\ref{prop:ent-frac} in the general case. 

\section{Explicit expression for $\vartheta_{\mu}(KG_{n,k})$}
In this section, we derive an explicit expression for the $\mu$-approximate Lov\'asz theta of the graph $KG_{n,k}$. Throughout this section, we assume that \(n\geq2k\). First we prove a variational form of $\vartheta_\mu(G)$. In the following, by $\lambda_{max}(A)$ and $\lambda_{min}(A)$, we denote the largest and the smallest eigenvalues of the matrix $A$ respectively.

\begin{lemma}[Variational form of $\vartheta_\mu(G)$]
    For any graph $G$ on $N$ vertices and any $\mu\geq 0$,
    \begin{align}
        \vartheta_\mu(G) = \max \left\{\lambda_{max}(V): V \in \mathbb{C}^{N\times N},V\geq 0, V_{ii}=1, |V_{ij}|\leq \mu \;\forall \{i,j\}\in E \right\}. \label{eq:var-lov}
    \end{align}
\end{lemma}
\begin{proof}
    For the direction $(\leq)$, consider the optimal $\mu$-OR $\{\ket{v_i}\}$ with handle $\ket{\psi}$ for $\vartheta_{\mu}(G)$. Define $T = \sum_{i=1}^N \ketbra{v_i}{i}$, for a set of orthonormal vectors $\{\ket{i}\}_{i=1}^N$. Now the frame operator, $F = \sum_{i=1}^N \ketbra{v_i}{v_i} = TT^\dagger$ and the Gram matrix, $V = \sum_{i,j=1}^N \langle v_i | v_j \rangle \ketbra{i}{j} = T^\dagger T$. Hence, $F$ and $V$ share their non-zero spectrum.

    Now, $V$ is a feasible solution to the optimization problem on the right-hand side of Eq.~\eqref{eq:var-lov}, and 
    \begin{align*}
        \sum_{i=1}^N |\langle \psi | v_i \rangle|^2 = \bra{\psi} F\ket{\psi} \leq \lambda_{max}(F) = \lambda_{max}(V).
    \end{align*}

    For the other direction $(\geq)$, consider the optimum $V$ for the optimization on the right-hand side. Therefore, columns of $V^{1/2}$, denoted by $\{\ket{v_i}\}_{i=1}^N$ form a $\mu$-OR of the graph $G$. Furthermore, choosing the handle vector $\ket{\psi}$ to be the unit eigenvector of $F$ with eigenvalue $\lambda_{max}(F)$, we have that $\sum_{i=1}^N |\langle \psi | v_i \rangle|^2 = \lambda_{max}(F) = \lambda_{max}(V)$.  
\end{proof}

We give a lower bound on $\vartheta_{\mu}(KG_{n,k})$ by constructing an explicit $\mu$-OR for the graph.
\begin{lemma}[Lower bound on $\vartheta_{\mu}(KG_{n,k})$] For all $0\leq \mu \leq 1$,
\begin{align}\label{eq:low}
    \vartheta_{\mu}(KG_{n,k}) \geq {n-1 \choose k-1}\left(1+ \mu \frac{n-k}{k}\right)
\end{align}
\end{lemma}
\begin{proof}
Denote $\ket{\beta} = \sum_{i=1}^n \ket{i}$ and for any $S\subseteq_k[n]$, denote $\ket{\chi_S} = \sum_{j\in S} \ket{j}$, where $\{\ket{i}\}$ is the standard orthonormal basis of $\mathbb{R}^n$.  In the following we show that for every $a \geq 0$, the unit vectors
   \begin{align}
       \ket{v_S} = \frac{\ket{\chi_S} + a \ket{\beta}}{\sqrt{k + 2ak + a^2n}}, \qquad \ket{\psi} = \frac{1}{\sqrt{n}}\ket{\beta},
   \end{align}
   $\forall S\subseteq_k[n]$ form a $\mu(a)$-OR of $KG_{n,k}$ with 
   \begin{align}
       \mu(a) = \frac{2ak + a^2 n }{k+2ak + a^2 n}, \qquad\sum_{S} |\langle \psi| v_S\rangle|^2 = {n-1 \choose k-1}\left(1+ \mu(a) \frac{n-k}{k}\right).
   \end{align}
   
    For disjoint $k$-subsets, $S$ and $T$, we calculate
    \begin{align*}
        \langle v_S|v_T \rangle = \frac{2ak + a^2n}{k+2ak + a^2 n} = \mu(a).
    \end{align*}
    Moreover, we calculate that for all $S\subseteq_k [n]$
    \begin{align*}
        |\langle \psi| v_S\rangle|^2 &= \frac{(k/\sqrt{n} + a\sqrt{n})^2}{k+2ak + a^2 n} = \frac{(k+an)^2}{n(k+2ak + a^2 n)}\\
        &= \frac{k^2 + 2akn + a^2n^2}{n(k+2ak + a^2 n)}\\
        &= \frac{k(k+2ak + a^2 n)+(n-k)(2ak + a^2 n)}{n(k+2ak + a^2 n)}\\
        &= \frac{k}{n} + \mu(a) \frac{n-k}{n} = \frac{k}{n}\left(1+ \mu(a)\frac{n-k}{k}\right).
    \end{align*}
    Summing over the ${n \choose k}$-many $k$-subsets, we immediately obtain
    \begin{align}
        \sum_S |\langle \psi| v_S\rangle|^2 = {n \choose k}\frac{k}{n}\left(1+ \mu(a)\frac{n-k}{k}\right) = {n-1 \choose k-1} \left(1+ \mu(a)\frac{n-k}{k}\right).
    \end{align}
    As $a$ ranges over $[0,\infty)$, $\mu(a)$ ranges over $[0,1)$. For $\mu = 1$, choose all vectors $\ket{v_S}=\ket{\psi}$, the handle vector, to obtain $\vartheta_1(KG_{n,k}) = {n \choose k}$. Thus we prove the lower bound.
\end{proof}

Now we prove a matching upper bound. By $A$, we denote the adjacency matrix, i.e., the $N\times N$ symmetric matrix whose $(i,j)$-entry is $1$ if $\{i,j\}\in E$, and $0$ otherwise. A graph $G$ is called $D$-regular if every vertex of $G$ has exactly $D$ adjacent vertices.
\begin{lemma}
    Let $G$ be $D$-regular on $N$ vertices with $s = -\lambda_{min}(A) >0$. Then for all $\mu\geq 0$,
    \begin{align}
        \vartheta_\mu(G) \leq \frac{N}{D+s}(s+\mu D).
    \end{align}
\end{lemma}
\begin{proof}
    Let $V$ be the optimum for Eq.~\eqref{eq:var-lov}, and let $u$ be its top unit eigenvector: $\lambda_{max}(V) = u^\dagger V u$. Define the Hermitian matrix $X$ with entries $X_{ij} = \Bar{u}_i V_{ij} u_j$. Then $X \geq 0, \Tr X = \sum_i |u_i|^2 = 1, X_{ii} = |u_i|^2$. Denoting by $J$ the matrix with all entries equal to $1$, we have
    \begin{align}
        \langle J, X\rangle := \Tr(JX) = \sum_{i,j} \Bar{u}_i V_{ij} u_j = u^\dagger V u = \lambda_{max}(V).
    \end{align}
    Moreover, $\forall\{i,j\}\in E, \;|V_{ij}|\leq \mu \implies |X_{ij}| \leq \mu \sqrt{X_{ii}X_{jj}}$. For a fixed $c>0$, we can split $J = (J-cA) + cA$ , and write
    \begin{align}
       \vartheta_\mu (G)= \lambda_{max}(V) = \langle J- cA, X\rangle + c \langle A,X\rangle.
    \end{align}
    For the first term, note that $J- cA \leq \lambda_{max}(J-cA) \mathbb{I}$. Therefore, $\langle J- cA, X\rangle \leq \lambda_{max}(J-cA) \Tr X = \lambda_{max}(J-cA)$. For the second term, note that 
    \begin{align*}
        \langle A,X\rangle = \sum_{\{i,j\}\in E} X_{ij} + X_{ji} &= 2\sum_{\{i,j\}\in E} \Re(X_{ij})\\
        &\leq 2\sum_{\{i,j\}\in E} |X_{ij}|\\
        &\leq 2 \mu\sum_{\{i,j\}\in E}  \sqrt{X_{ii}X_{jj}}\\
        & \leq \mu \sum_{\{i,j\}\in E} X_{ii} + X_{jj}\\
        & = \mu \sum_{i} \deg(i) X_{ii} = \mu D.
    \end{align*}
    For regular graphs, $J$ and $A$ commute~\cite{Brouwer2011-vh}. In fact $JA = AJ = DJ$. Now the only non-zero eigenvalue of $J$ is $N$ corresponding to the eigenvector, $(1,1,\dots,1 )^T$, and it is obvious that $A (1,1,\dots,1 )^T = D (1,1,\dots,1 )^T$. Therefore, one eigenvalue of $J-cA$ is $N-cD$. All other eigenvalues of $J-cA$ are of the form $-c\lambda$ for some $\lambda\in \mathrm{spec}(A)$. As $A$ has trace zero, at least one such $\lambda$ is negative. Hence, with $s = -\lambda_{min}(A)$, we can write $\lambda_{max}(J-cA) = \max\{N-cD,cs\}$. Therefore, for all $c>0$,
    \begin{align*}
        \lambda_{max}(V) \leq \max\{N-cD,cs\} + c\mu D.
    \end{align*}
    In particular, if we choose $c = N/(D+s)$, we have
    \begin{align}
       \vartheta_\mu(G) = \lambda_{max}(V) \leq \frac{N}{D+s}\left(s+\mu D\right).
    \end{align}
    Thus we prove the claim.
\end{proof}
Applying the preceding result to $KG_{n,k}$ which has $N = {n \choose k}$ vertices, is $D$-regular with $D = {n-k \choose k}$, and $s = {n-k-1 \choose k-1}$~\cite{Godsil2001-xn}, we have the following. 
\begin{corollary}[Upper bound on $\vartheta_{\mu}(KG_{n,k})$] For all $0\leq \mu \leq 1$,
\begin{align}\label{eq:up}
    \vartheta_{\mu}(KG_{n,k}) \leq {n-1 \choose k-1}\left(1+ \mu \frac{n-k}{k}\right)
\end{align}
\end{corollary}
Comparing Equations~\ref{eq:low} and~\ref{eq:up}, we obtain the exact expression for the $\mu$-approximate Lov\'asz theta function with $\mu \in [0,1]$,
\begin{align}\label{eq:exact}
    \vartheta_{\mu}(KG_{n,k}) = {n-1 \choose k-1}\left(1+ \mu \frac{n-k}{k}\right)
\end{align}

\section{Lower bound on the simulation error: Quantifying Hastings' conjecture}
With the explicit upper and lower bounds for $\vartheta_\mu(KG_{n,k})$, we can give a quantitative lower bound on the \textit{simulation error}, i.e., the trace distance between the target channel and the simulated channel. 

First we quantify the distributed monogamy of entanglement, in the following sense. 
\begin{prop}\label{prop:quant-dist-monogamy}
    If in the setting of Proposition~\ref{prop:ent-frac} in the main text, we consider the case where any $k$-subset, $S$ can extract a state $\epsilon$-close to an EPR pair with $A$ (and a failure flag otherwise), with probability $\lambda_S = \lambda>\epsilon$, then
    \begin{align}
        \lambda \leq \frac{k}{n} + \epsilon\left(1+ \frac{2d}{(\lambda - \epsilon)(d-1)}\cdot \frac{n-k}{n} \right).
    \end{align}
\end{prop}
\begin{proof}
    In this case we may set $\lambda_S = \lambda_T = \lambda$ in Eq.~\eqref{lem:overlap}, obtaining
    \begin{align}
        f(\epsilon) = \frac{2d\epsilon}{(d-1)(\lambda - \epsilon)}.
    \end{align}
    Substituting the upper bound, Eq.~\eqref{eq:up}, into Eq.~\eqref{eq:lam-av}, we obtain
    \begin{align}
        \lambda \leq \frac{k}{n} + \epsilon\left(1+ \frac{2d}{(\lambda - \epsilon)(d-1)}\cdot \frac{n-k}{n} \right).
    \end{align}
    Thus we prove the claim.
\end{proof}

Now we are in a position to quantify the simulation error in Hastings' problem. Specifically, we ask: \textit{what is the minimum $\epsilon$ such that} \begin{align}\label{eq:choi-dist}
    \limsup_{n\to \infty} \frac{1}{2}\norm{\mathcal{J}_{\mathcal{N}_\lambda}- \mathcal{J}_{\mathcal{D}_n\circ\mathcal{N}_{\gamma}^{\otimes n}\circ\mathcal{E}_n}}_1 \leq \epsilon
\end{align}
for some $\gamma < \lambda \leq \frac{1}{2}$? Theorem~\ref{thm:hastings} rules out asymptotically vanishing error.
We now derive a uniform strictly positive lower bound in
terms of \(\lambda\), \(\gamma\), and the input dimension
\(d\).
\begin{theorem}
    Let $\gamma < \lambda \leq \frac{1}{2}$ and let $d$ be the input dimension. For every sequence of encoder--decoder pairs, $\{(\mathcal{E}_n,\mathcal{D}_n)\}$, we have
    \begin{align}
       \liminf_{n\to \infty} \frac{1}{2}\norm{\mathcal{J}_{\mathcal{N}_\lambda}- \mathcal{J}_{\mathcal{D}_n\circ\mathcal{N}_{\gamma}^{\otimes n}\circ\mathcal{E}_n}}_1 \geq \lambda - \frac{1}{2}\left(\sqrt{
            \left[\frac{2d}{d-1}(1-\gamma)-\gamma\right]^2+ 4\lambda \frac{2d}{d-1}(1-\gamma)} - \left[\frac{2d}{d-1}(1-\gamma)-\gamma\right] \right)
    \end{align}
\end{theorem}
    \begin{proof}
        Pick any rational $t = q/(p+q) \in (\gamma,\lambda)$. Note the fact that $\mathcal{N}_t$ is $(p,q)$-extendible and $\mathcal{N}_\gamma$ can be obtained from $\mathcal{N}_t$ by further processing. As $(p,q)$-extendibility is monotonic under post processing, we conclude that $\mathcal{J}_{\mathcal{D}_n\circ\mathcal{N}_{\gamma}^{\otimes n}\circ\mathcal{E}_n}$ is $(p,q)$-extendible. For each blocklength $n$, define
\begin{align}
    \epsilon_n
    :=
    \frac{1}{2}
    \norm{
    \mathcal{J}_{\mathcal{N}_\lambda}
    -
    \mathcal{J}_{\mathcal{D}_n\circ
    \mathcal{N}_{\gamma}^{\otimes n}\circ\mathcal{E}_n}
    }_1 .
\end{align}
If $\epsilon_n\geq\lambda$ then the desired lower bound is immediately found. Hence suppose that $\lambda > \epsilon_n$. 
        
        Applying Proposition~\ref{prop:quant-dist-monogamy} gives
        \begin{align}
            \lambda  \leq \frac{q}{p+q}  + \epsilon_n\left(1+ \frac{2d}{(d-1)(\lambda - \epsilon_n)}\frac{p}{p+q}\right).
        \end{align}
        Denote $c = 2d/(d-1)$ and define $g_t(x)
:=
t+x\left(1+\frac{c(1-t)}{\lambda-x}\right),
0\le x<\lambda.
$
Then
$$\lambda\le g_t(\epsilon_n).$$
        As $1/2>t>0, 4\geq c>2$, note that $g(x)$ is a strictly increasing function of $x$ for $0 \leq x < \lambda$, with $g(0) = t$. So we solve for $x$ in the equation for $g(x) = \lambda$ to find the tightest lower bound on $\epsilon_n$.
        \begin{align*}
            &t + x \left(1+ \frac{c}{\lambda - x}(1-t)\right) = \lambda \\
            \implies &t(\lambda - x) + x(\lambda-x) + x c (1-t) = \lambda(\lambda - x)\\
            \implies &  x^2 -(2\lambda-t+c(1-t))x + (\lambda^2 - t\lambda) = 0\\
            \implies &x_{\pm} = \frac{2\lambda-t+c(1-t) \pm \sqrt{(2\lambda-t+c(1-t))^2 - 4\lambda(\lambda-t)}}{2} 
        \end{align*}
        Both of the roots are strictly positive. Simplifying the discriminant and noting that $x_+ >\lambda$, outside our regime of interest, we get the following lower bound for every blocklength,
        \begin{align}\label{eq:reuse}
            \epsilon_n \geq x_{-} = \lambda - \frac{1}{2}\left(\sqrt{[c(1-t)-t]^2+ 4\lambda c(1-t)} -[c(1-t)-t] \right).
        \end{align}
        Taking the sequence of rationals $t$ converging to $\gamma$, we find that the lower bound on $\epsilon_n$ converges to 
        \begin{align}
            \epsilon_n\geq \lambda - \frac{1}{2}\left(\sqrt{
            \left[\frac{2d}{d-1}(1-\gamma)-\gamma\right]^2+ 4\lambda \frac{2d}{d-1}(1-\gamma)} - \left[\frac{2d}{d-1}(1-\gamma)-\gamma\right] \right).
        \end{align}
        Finally in the asymptotically large blocklength limit,
        \begin{align}
           \liminf_{n\to\infty} \epsilon_n\geq \lambda - \frac{1}{2}\left(\sqrt{
            \left[\frac{2d}{d-1}(1-\gamma)-\gamma\right]^2+ 4\lambda \frac{2d}{d-1}(1-\gamma)} - \left[\frac{2d}{d-1}(1-\gamma)-\gamma\right] \right),
        \end{align}
        we finish the proof.
    \end{proof}    
        By taking the derivative of $x_-$ with respect to $c$, one checks that $x_-$ is a decreasing function of $c$. Hence we obtain the weakest bound at $c=4$ (equivalently at $d=2$)
        \begin{align}
            \liminf_{n\to \infty} \frac{1}{2}\norm{\mathcal{J}_{\mathcal{N}_\lambda}- \mathcal{J}_{\mathcal{D}_n\circ\mathcal{N}_{\gamma}^{\otimes n}\circ\mathcal{E}_n}}_1 \geq \lambda - \frac{1}{2}\left(\sqrt{
            \left(4-5\gamma\right)^2+ 16\lambda(1-\gamma)} -  \left(4-5\gamma\right) \right).
        \end{align}

\section{The case when $\lambda_S \leq \epsilon$ for some $k$-subset}
The proof of Proposition~\ref{prop:ent-frac} in the main text is demonstrated for the case when $\lambda_S> \epsilon$ for all $S\subseteq_k[n]$. Even though our proof of Hastings' conjecture only concerns the case: $\lambda_S = \lambda >\epsilon$ for all $S\subseteq_k[n]$, for the sake of completeness we include a proof that encompasses the case of $\lambda_S \leq \epsilon$.

Set the threshold $\tau:=\sqrt{\epsilon} \geq\epsilon$, and define $\mathcal{G}_{\tau}$, the set of good $k$-subsets, $S$, for which $\lambda_S > \tau$. For such sets $p_S \geq \lambda_S -\epsilon \geq. \lambda_S-\tau >0$, i.e., the unit vectors, $\ket{v_S}:= \frac{M_S\ket{\tilde{\sigma}}}{\sqrt{p_S}}$ are well-defined. Furthermore, for any $S,T\in \mathcal{G}_\tau$, we have 
\begin{align}
    |\langle v_S |v_T \rangle| \leq \frac{1}{\sqrt{(\lambda_S-\epsilon)(\lambda_T-\epsilon)}}\frac{2d\epsilon}{d-1} \leq  \frac{1}{\tau - \epsilon}\frac{2d\epsilon}{d-1} = \frac{1}{1 - \sqrt{\epsilon}}\frac{2d\sqrt{\epsilon}}{d-1}=: f(\epsilon)
\end{align}

When $S\notin \mathcal{G}_\tau$, we define unit vectors orthogonal to the vectors for the good sets, by enlarging the Hilbert space. Explicitly, when $S\in \mathcal{G}_\tau$,
\begin{align}
    \ket{v_S} = \frac{M_S\ket{\tilde{\sigma}}}{\sqrt{p_S}} \oplus 0,
\end{align} 
and when  $S\notin \mathcal{G}_\tau$,
\begin{align}
    \ket{v_S} = 0 \oplus \ket{e_S}.
\end{align}
Choose the vectors
$\{|e_S\rangle:S\notin\mathcal G_\tau\}$
to be mutually orthonormal. Therefore, the set of unit vectors $\{\ket{v_S}\}_{S\subseteq_k[n]}$ forms a $f(\epsilon)$-OR of $KG_{n,k}$. Choosing the handle vector, $\ket{\psi} = \ket{\tilde{\sigma}}\oplus 0$, we obtain by the definition of the approximate Lov\'asz theta function that
\begin{align}
    \sum_{S\subseteq_k[n]} |\langle \psi|v_S\rangle|^2 = \sum_{S\in \mathcal{G}_\tau} p_S \leq \vartheta_{f(\epsilon)}(KG_{n,k}).
\end{align}
Now, we write
\begin{align}
    \sum_{S\subseteq_k[n]} \lambda_S = \sum_{S\in \mathcal{G}_\tau} \lambda_S  + \sum_{S\notin \mathcal{G}_\tau} \lambda_S.
\end{align}
The first term is upper bounded as
\begin{align}
    \sum_{S\in \mathcal{G}_\tau} \lambda_S \leq \sum_{S\in \mathcal{G}_\tau} p_S + {n \choose k}\epsilon,
\end{align}
while the second term is upper bounded as
\begin{align}
    \sum_{S\notin \mathcal{G}_\tau} \lambda_S \leq {n \choose k}\tau = {n \choose k}\sqrt{\epsilon}.
\end{align}

Combining everything we have,
\begin{align}
    \sum_{S\subseteq_k[n]} \lambda_S \leq \vartheta_{f(\epsilon)}(KG_{n,k}) + {n \choose k}(\epsilon+\sqrt{\epsilon}) \leq  {n-1 \choose k-1}\left(1+ f(\epsilon) \frac{n-k}{k}\right)+{n \choose k}(\epsilon+\sqrt{\epsilon}).
\end{align}
where we have used the explicit upper bound in Eq~\eqref{eq:up}. Dividing both sides by ${n \choose k}$, we have
\begin{align}
    {n \choose k}^{-1}\sum_{S\subseteq_k[n]} \lambda_S \leq \frac{k}{n} + \epsilon+\sqrt{\epsilon} + \frac{1}{1 - \sqrt{\epsilon}}\left(\frac{2d\sqrt{\epsilon}}{d-1}\right)\frac{n-k}{n} = \frac{k}{n} + \Delta_\epsilon.
\end{align}
Note that $\lim_{\epsilon\to 0 }\Delta_{\epsilon}=0$. Thus we complete the proof.

\end{document}